\documentclass[a4paper, 11pt]{amsart}
\usepackage{enumerate}
\newcommand{\1}{{\tt b}}
\newcommand{\0}{{\tt a}}
\newcommand{\abs}[1]{|{#1}|}

\newcommand{\LL}{{\tt L}}
\newcommand{\ST}{{\tt S}}
\newcommand{\EE}{{\tt E}}
\newcommand{\cc}{{\tt c}}
\newcommand{\ff}{{\tt f}}
\newcommand{\DD}{{\tt D}}
\newcommand{\less}{\vartriangleleft}

\newtheorem{theorem}{Theorem}
\newtheorem{lemma}{Lemma}

\theoremstyle{remark}

\newtheorem{example}{Example}

\usepackage{amsmath, amssymb, amsthm}
\usepackage{tikz}
\usepackage{enumerate}
\usepackage{hyperref}
\usepackage[norelsize,algoruled,vlined,algo2e]{algorithm2e}
\usetikzlibrary{arrows}

\begin{document}
\title{Prefix frequency of lost positions}

\author{\v St\v ep\' an Holub}
\address[\v{S}. Holub]{Department of Algebra, Charles University, Sokolovsk\'a 83, 175 86 Praha, Czech Republic}
\email{holub@karlin.mff.cuni.cz}

\subjclass{68R15}
\keywords{runs, Lyndon words, periodicity}
\thanks{Supported by the Czech Science Foundation grant number 13-01832S}

\begin{abstract}
The concept of ``lost positions'' is a recently introduced tool for counting the number of runs in words.
We investigate the frequency of lost positions in prefixes of words. This leads to an algorithm that allows to show, using an extensive computer search, that the asymptotic density of runs in binary words is less than  $183/193\approx 0.9482$.   
\end{abstract}
\maketitle

\section{Periods, runs and Lyndon words}
In this paper we investigate the number of \emph{runs} in words over the binary alphabet $\Sigma=\{\0,\1\}$. A word $w$ here %(finite or infinite) 
is understood as a sequence $w[i]\in \Sigma$, $i=1,2\dots,n$, written $w=w[1]w[2]\cdots w[n]$. 
The integer $n$ is called the \emph{length} of $w$ and is denoted as $|w|$. We shall deal with \emph{factors} of a word $w$, which are words $w[i]w[i+1]\cdots w[j]$ for some $1\leq i \leq j\leq |w|$, abbreviated $w[i..j]$, with the convention $w[i..i]=w[i]$. We shall be also interested in integer intervals $[i..j]$ themselves, and also with individual integers $1,2,\dots,n$, which will be called \emph{positions} of $w$ in the appropriate context.  Words are naturally endowed with the operation of concatenation. If $u$ and $v$ are words, then their concatenation is written simply as $uv$.
 
We shall work with two basic properties of words, namely periods and lexicographic orders.
An integer $p$ is \emph{a period} of $w$ if $w[i+p]=w[i]$ for all $i=1,2,\dots,n-p$. \emph{The period} is the least period of the word. A word is \emph{primitive} if it is not a power of a shorter word. Every word $w$ is a power of a primitive word, which is known to be unique and is called the \emph{primitive root} of $w$. A basic result in combinatorics on words claims that two words $u$ and $v$ commute (that is, $uv=vu$) if and only if they have the same primitive root. Two words $w$ and $w'$ are called \emph{conjugate} if $w=uv$ and $w'=vu$ for some words $u$ and $v$.

The main object of research in this paper, a \emph{run} of $w$, is defined as an interval $[i..j]$, $1\leq i<j\leq |w|$, whose length $j-i+1$ is at least $2p$, where $p$ is the (least) period of $w[i..j]$, and such that the period $p$ cannot be extended, that is, $i$ is either $1$ or $w[i-1]\neq w[i-1+p]$, and $j$ is either $|w|$ or $w[j+1]\neq w[j+1-p]$. We remark that if $q$ is another period of $w[i..j]$ satisfying $j-i+1\geq 2q$, then $q$ is a multiple of $p$. This is a consequence of the Periodicity lemma (the discrete variant of Fine and Wilf theorems, see \cite{FW1965}), which claims that if a word $w$ of length at least $p+q-\gcd(p,q)$ has periods $p$ and $q$, then $\gcd(p,q)$ is a period of $w$ too.   

Interest in runs, also called ``maximal repetitions'', is part of a more general research on (long) repetitions which play an important role in the string processing, for example compression (see  \cite{crochemore09:_repet,smyth13, Crochemore_2014} for surveys). 
The maximal possible number of runs in a (binary) word of length $n$, denoted as $\rho(n)$, has been investigated a lot in recent years. 
Kolpakov and Kucherov  showed in \cite{KK} that $\rho(n) = O(n)$ which started the quest for the optimal constant, and also for other properties of the sequence $\rho(n)/n$. In a series of papers (\cite{rytter, puglisi, crochemore1, giraud, crochemore2}) the constant was progressively lowered towards $1$, often with a heavy computer computation. At first, runs were counted by their starting positions, later it turned out that using the center is more efficient. A breakthrough came with \cite{runsTheorem} where it became clear that decisively better choice is to consider the point of the run in which a particular Lyndon root starts. The consequence is a remarkably simple proof that $1$ is a strict upper bound on the constant. This result was expanded and accompanied by a new computer search in \cite{beyond}. The present paper develops the original idea a bit further and pushes the upper bound down by considering prefix density of \emph{lost positions}.    

The mentioned key concept of a Lyndon word is related to the second basic property of words, lexicographic orders. Let $\less$ be a lexicographic order on words. In our case of the binary alphabet, we have two orders $\less_c$, $c\in \Sigma$, defined by $\0\less_\0 \1$ and $\1\less_\1 \0$. A word $w$ is $\less$-Lyndon if for any factorization $w=uv$ (where $u$ and $v$ are not empty) we have $w\neq vu$ and $w\less vu$. The basic property of a Lyndon word is the following.
\begin{lemma}\label{border}
    Let $w=uvu$ be $\less$-Lyndon word. Then $u$ is empty.
\end{lemma}
\begin{proof}
The definition of Lyndon word yields $uvu\less uuv$ which implies $vu\less uv$, and $vuu\less uvu$. The claim follows from the definition of a Lyndon word.
\end{proof}
 The lemma says that a Lyndon word $w$ has no period shorter than $|w|$, in particular, it is primitive. On the other hand, any primitive word has a Lyndon conjugate: it is enough to take the $\less$-minimum of all conjugates. The relation to runs is given by the fact that if $[i..j]$ is a run, then $w[i..j]$ contains all conjugates of its period, or more precisely any conjugate of any factor of length $p$, where $p$ is the period of $w[i..j]$. If $w$ is a word, then a factor of $w$ which is $\less$-Lyndon is called the \emph{$\less$-Lyndon root} of $w$. Therefore, if $[i..j]$ is a run, then $w[i..j]$ contains each of the two Lyndon roots at least once. We give two more technical properties of Lyndon words for future reference.

\begin{lemma}\label{broken2}
    Let $z$ be $\less_\0$-Lyndon word, let $z' \1$ be a prefix of $z$, let $k\geq 1$ and let $u$ be an arbitrary word. 
    Then $z$ is the longest Lyndon word which is a prefix of  $z^kz'  \0u$. 
\end{lemma}
\begin{proof}
Let $w$ be a prefix of $z^kz'  \0u$ longer than $z$. If $w$ is a prefix of $z^kz'$, then $w$ has a period $z$ and therefore it is of the form $uvu$, hence not Lyndon by Lemma \ref{border}. On the other hand, if $w=z^kz'\0u'$, then $z'\0u'z^k \less_\0 w$ since $z'\1$ is a prefix of $w$. Therefore, again, $w$ is not a  Lyndon word.
\end{proof}

\begin{lemma}\label{broken}
    Let $z$ be  $\less_\0$-Lyndon word, let $z' \0$ be a prefix of $z$ and let $k\geq 1$. 
    Then $z^kz'  \1$ is  $\less_\0$-Lyndon word. 
\end{lemma}
The proof of Lemma \ref{broken} is a bit technical and we omit it (see \cite{fully2016}).  
\medskip

Parallel to the research on upper bounds, also several constructions of words with high density $\rho(n)/n$ of runs have appeared (\cite{franek,matsubara,jamie}), establishing lower bound for the sequence. By the end of this paper, we shall briefly compare our findings with these lower bounds.   
\medskip

The paper has two main parts. In Section \ref{theory}, we explain the theory that is behind the algorithm used in the computation. We also tackle the question of convergence of the sequence $\rho(n)/n$. Section \ref{sec3} yields some information about the computer computation and reports its results, which show that $\lim_{n\to \infty}\rho(n)/n<183/193$.
 
\section{Theory}\label{theory}

\subsection{Lost positions}
In this section we explain the classification of positions according to their relation to runs. This is the fundamental tool introduced in \cite{runsTheorem} and developed further in \cite{beyond}.

If $c\in \Sigma$, we shall denote as $\overline c$ the other letter from $\Sigma$, distinct from $c$. 
Let $1<i\leq |w|$ be a position in  a word $w$. We assign to $i$ the following parameters. 
\begin{align*}
\cc_w(i)&=\overline{w[i-1]}, \\
\LL_w(i)&=\max\{j\ |\ \text{$w[i..j]$ is  $\less_{\cc_w(i)}$-Lyndon word}\}, \\
\DD_w(i)&=\LL_w(i)-i+1, \\
\ST_w(i)&=\min\{j \ |\ \text{$\DD_w(i)$ is a period of $w[j..\LL_w(i)]$}\}, \\
\EE_w(i)&=\max\{j \ |\ \text{$\DD_w(i)$ is a period of $w[i..j]$}\}.
\end{align*}
By definition, the interval $[\ST_w(i)..\EE_w(i)]$  is the maximal extension of the interval $[i..\LL_w(i)]$ that has a period $\DD_w(i)$. The word $w[i..\LL_w(i)]$ is a Lyndon root of $w[\ST_w(i)..\EE_w(i)]$. 

\begin{example}\label{ex}
Parameters are illustrated by Table \ref{example} for the word \[w=\0\0\1\0\1\1\1\0\1\0\0\1\0\1\0\0\0.\] Take for example the position $10$. Since $w[10]=\0$ and $w[9]=\1$, we are looking for the longest Lyndon word with respect to $\less_\0$ starting at the position $10$. 
Such is the word $w[10..14]=\0\0\1\0\1$. It can be extended to both sides to $w[7..16]=\1\0\1\vert\0\0\1\0\1\vert\0\0$. The interval $[7..16]$ is a run in $w$ and $\0\0\1\0\1$ is the $\less_\0$-Lyndon root of $w[7..16]$. 

The $\less_\1$-Lyndon root of $w[7..16]$ is $\1\0\1\0\0$ and it occurs twice in $w[7..16]$, namely at positions $7$ and $12$. However, this particular Lyndon word is not considered in parameters of either $7$ or $12$. Since $w[6]=\1$, the position $7$ asks for $\less_\0$-Lyndon word. Largest such a word starting at the position $7$ is the single letter $\1$. Note that this is the case for all positions where $w[i]=w[i-1]$. More precisely $w[i]=w[i-1]$ implies $i=\LL_w(i)$.

As for the position $12$, here $\1\0\1\0\0$ is a Lyndon word with respect to the required order, but it is not the longest one starting at $12$. Extension to the right by one position yields a longer Lyndon word $\1\0\1\0\0\0$. It is useful to see the fact that $\1\0\1\0\0\0$ is a Lyndon word as a consequence of Lemma \ref{broken} (with the role of the letters interchanged). Informally, if the ``Lyndon period'' is ``broken'' by the lexicographically greater letter, then we obtain a new Lyndon word.  
\end{example}

\begin{table}[!htb]
  \caption[]{Parameters of positions in the word $\0\0\1\0\1\1\1\0\1\0\0\1\0\1\0\0\0$ }
  \label{example}
\begin{tikzpicture}[scale=0.5,inner sep=1pt]
\def\slovo{\0/1,\0/2,\1/3,\0/4,\1/5,\1/6,\1/7,\0/8,\1/9,\0/10,\0/11,\1/12,\0/13,\1/14,\0/15,\0/16,\0/17}
\draw[help lines] (-1,0) grid (17,-6);
\node at (-0.5,-1.5) {\tiny $i$};
\node at (-0.5,-2.5) {$\LL$}; 
\node at (-0.5,-3.5) {$\DD$}; 
\node at (-0.5,-4.5) {$\ST$};
\node at (-0.5,-5.5) {$\EE$};  
\foreach \p in {-2,-3,-4,-5}
\node at (0.5,\p-0.5) {$\times$}; 
\foreach \l/\p in \slovo
\node at (\p-0.5,-0.5) {$\l$}; 
\foreach \p in {1,2,...,17}
\node at (\p-0.5,-1.5) {\tiny \p}; 
%L
\foreach \p/\l in 
{2/2,3/4,4/7,5/17,6/6,7/7,8/9,9/11,10/14,11/11,12/17,13/14,14/17,15/15,16/16,17/17}
\node at (\p-0.5,-2.5) {\small \l}; 
%D
\foreach \p/\l in 
{2/1,3/2,4/4,5/13,6/1,7/1,8/2,9/3,10/5,11/1,12/6,13/2,14/4,15/1,16/1,17/1}
\node at (\p-0.5,-3.5) {\small \l}; 
%S
\foreach \p/\l in 
{2/1,3/2,4/3,5/4,6/5,7/5,8/7,9/8,10/7,11/10,12/10,13/11,14/13,15/15,16/15,17/15}
\node at (\p-0.5,-4.5) {\small \l}; 
%E
\foreach \p/\l in 
{2/2,3/5,4/9,5/17,6/7,7/7,8/10,9/13,10/16,11/11,12/17,13/15,14/17,15/17,16/17,17/17}
\node at (\p-0.5,-5.5) {\small \l}; 
\end{tikzpicture}
\end{table}

Consider now the mapping
\[\+R_w: i \mapsto [\ST_w(i)..\EE_w(i)]\,, \]
defined for $i=2,3,\dots,|w|$. 
Let $[s..e]$ be in the range of $\+R_w$, and let us investigate positions $i$ for which $\+R_w(i)=r$. Candidates are the starting positions of Lyndon roots of $[s..e]$. These form two arithmetic progressions with the common difference $p$, where $p$ is the period of $w[s..e]$. For example, the word $\0\1\1\0\0\1\1\0\0\1\1\0\0\1$ has the period 4, and Lyndon roots $\0\0\1\1$ and $\1\1\0\0$ with sequences of starting positions $(4,8)$ and $(2,6,10)$ respectively (see Table \ref{progres}). 

\begin{table}%
\caption{Arithmetic progressions of Lyndon roots}
\label{progres}
\begin{tikzpicture}
	[scale=0.5,inner sep=1pt]
\def\slovo{\0/1,\1/2,\1/3,\0/4,\0/5,\1/6,\1/7,\0/8,\0/9,\1/10,\1/11,\0/12,\0/13,\1/14}
\draw[help lines] (0,0) grid (14,1);
\foreach \x in {2,6,10}
\draw node[fill=white] at (\x-0.5,1) {{\tiny \x}};
\foreach \x in {4,8}
\draw node[fill=white] at (\x-0.5,0) {{\tiny \x}};
\foreach \l/\p in \slovo
\node at (\p-0.5,0.5) {$\l$};
\foreach \p in {0,1,2}
\draw[gray] (1+4*\p,1) .. controls (2+4*\p,2) and (4+4*\p,2) .. (5+4*\p,1); 
\foreach \p in {0,1}
\draw[gray] (3+4*\p,0) .. controls (4+4*\p,-1) and (6+4*\p,-1) .. (7+4*\p,0); 
\end{tikzpicture}
\end{table}

However, there are two requirements on these candidates: a) the word starting at $i$ must be $\less_{\cc_w(i)}$-Lyndon; and b) it must be the longest Lyndon word starting at $i$. 

We show that the condition a) is satisfied whenever $i=\LL_w(i)$ or $s<i$. In the case $i=\LL_w(i)$, the single letter $w[i]=w[i..\LL_w(i)]$ is trivially a Lyndon word for both orders.
 If $i<\LL_w(i)$ and $s<i$, then the Lyndon word $w[i..\LL_w(i)]$ contains both letters, hence it is $\less_{w[i]}$-Lyndon. The period of $w[s..e]$ implies $w[i-1]=w[\LL_w(i)]$, and Lemma \ref{border} yields $w[i]\neq w[\LL_w(i)]$. Therefore  $w[i]=\cc_w(i)$.

The condition b) is satisfied if $e=|w|$, by Lemma \ref{border} (a Lyndon word has no period shorter than itself). If $e<|w|$, then Lemma \ref{broken2} and Lemma \ref{broken} imply that b) is satisfied if and only if $w[e+1]=\cc_w(i)$. 

Altogether, we have obtained the following characterization of preimages of $r=[s..e]$: $\+R(i)=r$ if and only if either
 \begin{enumerate}[(A)]
	\item $e<|w|$ and $s<i$, and $i$ is the starting position of the $\less_{w[e+1]}$-Lyndon root of $w[s..e]$; or \label{caseA}
	\item $e=|w|$ and $s<i$, and $i$ is the starting position of  any Lyndon root of $w[s..e]$; or \label{caseB}
	\item $r=[i..|w|]$ and $w[r]$ has the period one. \label{caseC}
\end{enumerate}
Note that the case \eqref{caseC} applies to a single position, namely to the starting position of the last ``block'' of letters in $w$.

It is now straightforward to observe that each run of $w$ is in the range of $\+R_w$. This is the groundbreaking observation of \cite{runsTheorem} which immediately proves that $\rho(n)<n$ (the inequality is strict since the first position is not mapped). On the other hand, $\+R_w$ need not be injective, and $\+R_w(i)$ may not be a run. These two options constitute a basis for further lowering of $\rho(n)$.

In Table \ref{example},  positions $6$ and $7$ are an example of non-injectivity. And the position $4$ is mapped to $[3..9]$ of length seven which is not a run since the least period of $w[3..9]$ is four. This alone implies that the word in the table has at most $|w|-3$ runs. We are going to say that positions $6$ and $4$ are lost.

Given a run $r=[s..e]$, we pick a representative $\+C_w(r)$ of positions mapped to $r$ as follows. If $e<|w|$ %or if $r$ is the last block of letters (formally, $e=|w|$ and the period of $r$ is one)
, then $\+C_w(r)$ is the maximum of all positions $i$ such that $\+R_w(i)=r$.

If $e=|w|$, then there are two progressions one for each lexicographic order (they coincide if the period of $[s..e]$ is one). In this case we pick the last member of the progression that starts later. 
More formally, let 
\begin{align*}
r_w(c)=\min\{i\ |\ \text{$\+R_w(i)=r$ and $w[i..\LL_w(i)]$ is $\less_c$-Lyndon}\},
\end{align*}
be the starting elements of the two progressions. We define $\ff_w(r)$ to be such that $r_w({\ff_w(r)})\geq r_w({\overline{\ff_w(r)}})$ (we choose $\ff_w(r)=\0$ if the period of $[s..e]$ is one). 
Then
\begin{align*}
\+C_w(r)=
\max\{i\ |\ \text{$\+R_w(i)=r$ and $w[i..\LL_w(i)]$ is $\less_{\ff_w(r)}$-Lyndon}\}\,.
\end{align*}
Using the the word in Table \ref{progres}, we have $\ff_w([1,14])=\0$ and $\+C_w([1,14])=8$.
%Note also that if the period of $r$ is one, the two sequences coincide. We adopt the convention $\ff_w(r)=\0$ for such a case.

We now say that the position $1<i\leq |w|$ in $w$ is 
\begin{enumerate}[(a)]
	\item \emph{charged} if $i=\+C_w(r)$ for some run $r$;
	\item \emph{right open} if $\EE_w(i)=|w|$, and
		\begin{enumerate}[(ba)]
		\item either $\+R_w(i)$ is not a run, or \label{ba}
		\item $\+R_w(i)$ is a run and $w[i..\LL_w(i)]$ is not $\less_{\ff_w(r)}$-Lyndon, or \label{bb}
		\item $i=\ST_w(i)$; \label{bc}
		\end{enumerate}
	\item \emph{left open} if  $\ST_w(i)=1$ and $\+R_w(i)$ is not a run; \label{leftopen}	
	\item \emph{lost} otherwise.	
\end{enumerate}
Since we will be most interested in lost positions, let us give their positive characterization. The position $i$ is lost in $w$ if and only if
\begin{enumerate}[(i)]
	\item $\EE_w(i)<|w|$ and $\ST_w(i)>1$ and $\+R_w(i)$ is not a run, or \label{lost1}
	\item $\EE_w(i)<|w|$, $\+R_w(i)$ is a run, and $i+\DD_w(i)\leq \EE_w(i)$, or \label{lost2}
	\item $\EE_w(i)=|w|$, $\+R_w(i)$ is a run $r$, $i+\DD_w(i)\leq \EE_w(i)$, $\ST_w(i)<i$ and $w[i..\LL_w(i)]$ is $\less_{\ff_w(r)}$-Lyndon.\label{lost3}
\end{enumerate}

Conspicuous complications when it comes to positions with $\EE_w(i)=|w|$ (or, to a lesser extent, $\ST_w(i)=1$) are motivated by our desire to detect as many lost positions as possible and in the same time to make the definitions compatible with extensions of the word $w$ to the left and right. All previously defined parameters of $i$ remain unchanged when we start to consider a word $w_1ww_2$ instead of $w$ when $\ST_w(i)>1$ and $\EE_w(i)< |w|$. If $\EE_w(i)=|w|$, then the parameters may change as Table \ref{example2} illustrates for the word from Table \ref{example} appended with $\0$ or with $\1$. If   $\EE_w(i)< |w|$ and $\ST_w(i)=1$, then only $\ST_{w_1ww_2}(i)$ can be different.

\begin{table}[!htb]
  \caption[]{Changed parameters}
  \label{example2}
\begin{center}
\begin{tikzpicture}[scale=0.5,inner sep=1pt]
\def\slovo{\0/1,\0/2,\1/3,\0/4,\1/5,\1/6,\1/7,\0/8,\1/9,\0/10,\0/11,\1/12,\0/13,\1/14,\0/15,\0/16,\0/17,\1/18}
\draw[help lines] (-1,0) grid (18,-6);
\foreach \x/\y in {5/4,12/4,14/4,15/1,15/3,15/4}
\node[fill=black!20, inner sep=3.5pt, style=circle] at (\x-0.5,-1-\y-0.5) {}; 
\draw [fill=black!10,black!10] (17,0) rectangle (18,-6);
\node at (-0.5,-1.5) {\tiny $i$};
\node at (-0.5,-2.5) {$\LL$}; 
\node at (-0.5,-3.5) {$\DD$}; 
\node at (-0.5,-4.5) {$\ST$};
\node at (-0.5,-5.5) {$\EE$};  
\foreach \p in {-2,-3,-4,-5}
\node at (0.5,\p-0.5) {$\times$}; 
\foreach \l/\p in \slovo
\node at (\p-0.5,-0.5) {$\l$}; 
\foreach \p in {1,2,...,18}
\node at (\p-0.5,-1.5) {\tiny \p}; 
%L
\foreach \p/\l in 
 {2/2,3/4,4/7,5/17,6/6,7/7,8/9,9/11,10/14,11/11,12/17,13/14,14/17,15/18,16/16,17/17,18/18}
\node at (\p-0.5,-2.5) {\small \l}; 
%D
\foreach \p/\l in 
 {2/1,3/2,4/4,5/13,6/1,7/1,8/2,9/3,10/5,11/1,12/6,13/2,14/4,15/4,16/1,17/1,18/1}
\node at (\p-0.5,-3.5) {\small \l}; 
%S
\foreach \p/\l in 
{2/1,3/2,4/3,5/4,6/5,7/5,8/7,9/8,10/7,11/10,12/10,13/11,14/13,15/13,16/15,17/15,18/18}
\node at (\p-0.5,-4.5) {\small \l}; 
%E
\foreach \p/\l in 
 {2/2,3/5,4/9,5/18,6/7,7/7,8/10,9/13,10/16,11/11,12/18,13/15,14/18,15/18,16/17,17/17,18/18}
\node at (\p-0.5,-5.5) {\small \l}; 
\node at (0,-6.5) {};
\end{tikzpicture}

\begin{tikzpicture}[scale=0.5,inner sep=1pt]
\def\slovo{\0/1,\0/2,\1/3,\0/4,\1/5,\1/6,\1/7,\0/8,\1/9,\0/10,\0/11,\1/12,\0/13,\1/14,\0/15,\0/16,\0/17,\0/18}
\draw[help lines] (-1,0) grid (18,-6);
\draw [fill=black!10,black!10] (17,0) rectangle (18,-6);
\foreach \x/\y in {5/1,5/2,5/4,12/1,12/2,12/4,14/1,14/2,14/4,15/4,16/4,17/4}
\node[fill=black!20, inner sep=3.5pt, style=circle] at (\x-0.5,-1-\y-0.5) {}; 
\node at (-0.5,-1.5) {\tiny $i$};
\node at (-0.5,-2.5) {$\LL$}; 
\node at (-0.5,-3.5) {$\DD$}; 
\node at (-0.5,-4.5) {$\ST$};
\node at (-0.5,-5.5) {$\EE$};  
\foreach \p in {-2,-3,-4,-5}
\node at (0.5,\p-0.5) {$\times$}; 
\foreach \l/\p in \slovo
\node at (\p-0.5,-0.5) {$\l$}; 
\foreach \p in {1,2,...,18}
\node at (\p-0.5,-1.5) {\tiny \p}; 
%L
\foreach \p/\l in 
 {2/2,3/4,4/7,5/18,6/6,7/7,8/9,9/11,10/14,11/11,12/18,13/14,14/18,15/15,16/16,17/17,18/18}
\node at (\p-0.5,-2.5) {\small \l}; 
%D
\foreach \p/\l in 
 {2/1,3/2,4/4,5/14,6/1,7/1,8/2,9/3,10/5,11/1,12/7,13/2,14/5,15/1,16/1,17/1,18/1}
\node at (\p-0.5,-3.5) {\small \l}; 
%S
\foreach \p/\l in 
{2/1,3/2,4/3,5/4,6/5,7/5,8/7,9/8,10/7,11/10,12/10,13/11,14/13,15/15,16/15,17/15,18/15}
\node at (\p-0.5,-4.5) {\small \l}; 
%E
\foreach \p/\l in 
{2/2,3/5,4/9,5/18,6/7,7/7,8/10,9/13,10/16,11/11,12/18,13/15,14/18,15/18,16/18,17/18,18/18}
\node at (\p-0.5,-5.5) {\small \l}; 
\end{tikzpicture}
\end{center}
\end{table}

We mainly care about differences regarding lost positions. We would like to claim that lost positions remain uncharged in any extension of $w$ (that is why we call them lost, they will never be charged). This is true for lost positions with $\EE_w(i)<|w|$ but not always for positions lost according to \eqref{lost3}. Why we call them lost, then? We explain the idea using the word $w=\0\1\1\0\0\1\1\0\0\1\1\0\0\1$ from Table \ref{progres}.
The run $[1..14]$ is fairly overloaded, it is the image of positions $\{2,4,6,8,10\}$ under $\+R_w$. We argue that at least one of those positions must be lost in any extension of $w$. More precisely, we argue that (at least) one position of the pair $\{2,4\}$ is lost in any extension but it is not clear which one. This depends on whether the period of $w[1..14]$ will be ``broken'' by $\0$ or by $\1$. The two possibilities are illustrated in Table \ref{extend}.
If $w_2=\1\0\1\0$ (and $w_1$ is empty), then we obtain
\begin{align*}
 \+R_{ww_2}(2)&=\+R_{ww_2}(6)=\+R_{ww_2}(10)=[1..16], \\
 \+R_{ww_2}(4)&=[3..18], \\
 \+R_{ww_2}(8)&=[7..18], \\
  \LL_{ww_2}(4)&=\LL_{ww_2}(8)=17,
\end{align*}
and $2$ (as well as $6$) is lost by \eqref{lost2}.
On the other hand, for $w_2=\0\1$, we have
\begin{align*}
 \+R_{ww_2}(4)&=\+R_{ww_2}(8)=[1..14], \\
 \+R_{ww_2}(2)&=[1..16], \\
 \+R_{ww_2}(6)&=[5..16], \\
\+R_{ww_2}(10)&=[9..16], \\
\LL_{ww_2}(2)&=\LL_{ww_2}(6)=\LL_{ww_2}(10)=15,
\end{align*}
and $4$ is lost by \eqref{lost2}.

\begin{table}%
\caption{Lost positions after extension}
\label{extend}
\begin{tikzpicture}
	[scale=0.5,inner sep=1pt]
\def\slovo{\0/1,\1/2,\1/3,\0/4,\0/5,\1/6,\1/7,\0/8,\0/9,\1/10,\1/11,\0/12,\0/13,\1/14,\1/15,\0/16,\1/17,\0/18}
\draw [fill=black!10,black!10] (0,0) rectangle (14,1);
\node[fill=black!20, inner sep=3.5pt, style=circle] at (16.5,0.5) {};
\draw[help lines] (0,0) grid (18,1);
\foreach \x in {2,6,10,14,16}
\draw node[fill=white] (n\x) at (\x-0.5,1) {{\tiny \x}};
\foreach \x in {4,8,18}
\draw node[fill=white] (n\x) at (\x-0.5,0) {{\tiny \x}};
\foreach \l/\p in \slovo
\node at (\p-0.5,0.5) {$\l$};
\begin{scope}
\clip (0,1) rectangle (18,2);
\foreach \p in {-1,0,1,2}
\draw[gray] (1+4*\p,1) .. controls (2+4*\p,2) and (4+4*\p,2) .. (5+4*\p,1); 
\end{scope}
%
%\foreach \p in {0,1}
\draw[gray] (3,0) .. controls (4,-3) and (16,-3).. (17,0); 
\draw[gray] (7,0) .. controls (8,-2) and (16,-2).. (17,0); 
\begin{scope}
\clip (17,0) rectangle (18,-1);
\draw[gray] (17,0) .. controls (18,-1.5) and (30,-1.5) .. (31,0); 
\end{scope}
\begin{scope}
\clip (6,0) rectangle (7,-1);
\draw[gray] (7,0) .. controls (6,-2) and (-2,-2) .. (-3,0); 
\end{scope}
\begin{scope}
\clip (3,0) rectangle (2,-1);
\draw[gray] (3,0) .. controls (2,-2) and (-13,-1.5) .. (-14,0); 
\end{scope}
\begin{scope}
\clip (13,1) rectangle (16,2);
\draw[gray] (13,1) .. controls (14,2) and (16,2) .. (17,1); 
\end{scope}
\begin{scope}
\clip (0,0) rectangle (16,-1);
\foreach \p in {-1,0,1,2,3}
\draw[dotted] (3+4*\p,0) .. controls (4+4*\p,-0.7) and (6+4*\p,-0.7) .. (7+4*\p,0); 
\end{scope}
\end{tikzpicture}
%######################
%######################
\begin{tikzpicture}
	[scale=0.5,inner sep=1pt]
\def\slovo{\0/1,\1/2,\1/3,\0/4,\0/5,\1/6,\1/7,\0/8,\0/9,\1/10,\1/11,\0/12,\0/13,\1/14,\0/15,\1/16}
\draw [fill=black!10,black!10] (0,0) rectangle (14,1);
\node[fill=black!20, inner sep=3.5pt, style=circle] at (14.5,0.5) {};
\draw[help lines] (0,0) grid (16,1);
\foreach \x in {2,6,10,14}
\draw node[fill=white] (n\x) at (\x-0.5,1) {{\tiny \x}};
\foreach \x in {4,8,12}
\draw node[fill=white] (n\x) at (\x-0.5,0) {{\tiny \x}};
\foreach \l/\p in \slovo
\node at (\p-0.5,0.5) {$\l$};
\begin{scope}
\clip (0,0) rectangle (14,-2);
\foreach \p in {-1,0,1,2}
\draw[gray] (3+4*\p,0) .. controls (4+4*\p,-1) and (6+4*\p,-1) .. (7+4*\p,0); 
\end{scope}
\foreach \p/\h in {1/3.5,5/2.5,9/2}
{
\begin{scope}
\clip (\p-1,1) rectangle (16,3);
\draw[gray] (\p,1) .. controls (\p+1,\h) and (14,\h) .. (15,1);
\draw[gray] (\p,1) .. controls (\p-1,\h) and (2*\p-14,\h) .. (2*\p-15,1);
\end{scope}
}
\begin{scope}
\clip (15,1) rectangle (16,3);
\draw[gray] (15,1) .. controls (16,3) and (22,3) .. (23,1);
\end{scope}
\begin{scope}
\clip (0,1) rectangle (14,3);
\foreach \p in {-3,1,5,9,13}
\draw[gray,dotted] (\p,1) .. controls (\p+1,1.5) and (\p+3,1.5) .. (\p+4,1);
\end{scope}
\end{tikzpicture}
\end{table}

We now formulate the correct version of the informal ``lost is lost forever'' precisely. 
\begin{lemma}\label{lem}
Let $w$, $w_1$ and $w_2$ be words, and let $p_1< p_2< \dots < p_k$ be all lost positions of a word $w$. Then $1<p_1$ and $p_k<|w|$, and there is a monotonically increasing injective mapping 
\[\mu: \{p_1, p_2,\dots,p_k\} \to \{1,2,\dots, |w|-1\} \]
such that, for all $i=1,2,\dots,k$, the position $|w_1|+\mu(p_i)$ is lost in $w_1ww_2$ and $\mu(p_i)\leq p_i$.
\end{lemma}
\begin{proof}
It is straightforward to see from definitions that  the first and the last positions of a word are not lost, therefore $1<p_1$ and $p_k<|w|$.

Denote $w'=w_1ww_2$, $p_i'=|w_1|+p_i$, $r=\+R_w(p_i)$, and $r'=\+R_{w'}(p_i')$. We first deal with four situations when $p_i'$ is lost in $w'$. For such positions we define $\mu(p_i)=p_i$. 

1. Conditions \eqref{lost1} and \eqref{lost2} above yield that if $\EE_w(p_i)<|w|$, then $p_i'$  is lost in $w'$.

2.  Let $\EE_w(p_i)=|w|$, $\DD_w(p_i)=\DD_{w'}(p_i')$ and $\EE_{w'}(p_i')<|w'|$. Since $p_i$ is lost, it satisfies \eqref{lost3} in $w$. Then \eqref{lost2} applies to $p_i'$ in $w'$.

3. Let $\EE_w(p_i)=|w|$, $\DD_w(p_i)=\DD_{w'}(p_i')$, $\EE_{w'}(p_i')=|w'|$ and $\ff_w(r)=\ff_{w'}(r')$. Then $p_i'$ satisfies \eqref{lost3} in $w'$.

4. Let $\EE_w(p_i)=|w|$ and $\DD_w(p_i)=1$. This case applies, informally stated, to inner positions of the last block of letters in $w$. It follows from definitions that inner positions of any block of letters are lost.  More formally, let $t$ be such that $w=u\overline c c^t$ or $w=c^t$ for a letter $c$. Since $p_i$ is lost, we deduce from \eqref{lost3} that $|w|-t+1<p_i<|w|$. (Observe that for $j=|w|-t+1$, we have $\ST_w(j)=j$.) We can either directly verify that the position $p_i'$ is lost in $w'$ or to note that we are in one of the previous cases.

The remaining cases are $\EE_w(p_i)=|w|$, $1\neq \DD_w(p_i)$, and either 
\begin{align*}
\DD_w(p_i)< \DD_{w'}(p_i'), \tag{$*$}\label{pos1}
\end{align*} 
or 
\begin{align*}
\begin{split}
&\DD_w(p_i)=\DD_{w'}(p_i'),\\ 
&\EE_{w'}(p_i')=|w'|, \quad \text{and} \\ 
&\ff_w(r)\neq \ff_{w'}(r').
\end{split}
\tag{$**$}
\label{pos2}
\end{align*}
Then we define $\mu(p_i)=p_i-\Delta$, where %$r=[\ST_w(i)..|w|]$ and
\[\Delta=r_w(\ff_w(r))-r_w({\overline {\ff_w(r)}}).\]
In other terms, $\Delta$ is the shift between the two arithmetic progression of starting positions of Lyndon roots of $w[\ST_w(i)..|w|]$ for the two lexicographic orders. We have $\Delta> 0$ and thus $\mu(p_i)<p_i$ (this motivates the way in which $\ff_w(r)$ was chosen). 
Since $p_i$ is lost in $w$, we deduce from \eqref{lost3} that $w[p_i..\LL_w(p_i)]$ is $\less_{\ff_w(r)}$-Lyndon, hence $w[p_i]=\ff_w(r)$. 
Thus $w[p_i-\Delta.. \LL_w(p_i-\Delta)]$ is $\less_{\overline{\ff_w(r)}}$-Lyndon. This implies that $\mu(p_i)$ is open in $w$ because it satisfies the condition \eqref{bb} and the injectivity of $\mu$ is not violated.
It remains to show that $j=|w_1|+\mu(p_i)$ is lost in $w'$. 

Consider first \eqref{pos1} and let $uc$, $c\in \Sigma$, be the shortest prefix of $w_2$ such that $w[p_i..|w|]uc$ has not the period $\DD_w(p_i)$. Since $\DD_w(p_i)\neq \DD_{w'}(p_i')$, Lemma \ref{broken2} applied to the word $w[p_i..|w|]uc$ implies that $c\neq {w[p_i]}$. Lemma \ref{broken2}, applied this time to the word $w[p_i-\Delta..|w|]uc$, implies $\DD_{w}(p_i-\Delta)=\DD_{w'}(j)$. Then $\EE_{w'}(j)=|w_1wu|<|w'|$, and $j$ is lost in $w'$ by the condition \eqref{lost2}.

The conditions \eqref{pos2} yield that $\DD_{w'}(j)=\DD_w(p_i-\Delta)$, and  $j$ is lost in $w'$ by \eqref{lost3} since $\less_{\ff_{w'}(r')}=\less_{\overline{\ff_w(r)}}$  and $w'[j..\LL_{w'}(j)]=w[p_i-\Delta..\LL_w(p_i-\Delta)]$ is $\less_{\overline{\ff_w(r)}}$-Lyndon.
Note that this case can happen if $\ST_w(p_i)=1$ and $\ST_{w'}(p_i')\leq |w_1|$.
\end{proof}

\subsection{Convergence}\label{conv} We are trying to give an upper bound on 
 \[\lim_{n\to\infty}\frac {\rho(n)}n\,.\]
The function $\rho(n)/n$ is not monotonically increasing, for example, $\rho(4)=\rho(5)=2$. It is therefore apriori not clear that the limit exists, and that we are not forced to use the upper limit instead. This difficulty is addressed in \cite{giraud} by a reference to the classical result on superadditive (or subadditive) sequences, known as Fekete's Lemma, which claims that if a sequence $(a_n)$ is superadditive, that is, if it satisfies $a_n+a_m\leq a_{n+m}$, then $\lim \frac {a_n}n$ exists (see \cite[p. 25]{fekete} for the discussion of the history of this result). 

As long as the alphabet size is not limited, the superadditivity of $\rho(n)$, that is, the inequality $\rho(m)+\rho(n)\leq \rho(m+n)$, is easily obtained in the following way. Let $x$ and $y$ be words of lengths $m$ and $n$ that contain $\rho(m)$ and $\rho(n)$ runs respectively. Then $xy$ contains exactly $\rho(m)+\rho(n)$ runs, given that the alphabets of $x$ and $y$ are disjoint. 

This simple argument of course cannot be used for a fixed alphabet. Purported remedy for a fixed alphabet is given by Proposition 5 of \cite{giraud}. There it seems to be claimed that the total number of runs in words $uw$ and $wv$, where $w$ is the longest word which is both suffix of $uw$ and a prefix of $wv$, is the same as the total number of runs in words $uwv$ and $w$. 
However, this is not true as the example $u={\tt bab}$, $w={\tt cabc}$ and $v={\tt abcb}$ shows. As a consequence, it seems to be an open question whether $\rho(n)$ is superadditive for a fixed alphabet.

In any case, the limit exists since a variant of Fekete's Lemma holds also for functions that satisfy a weaker form of the superadditivity condition (see \cite[p. 162, Theorem 23]{erdos} for details). Our case is a simple example of such a stronger form of Fekete's Lemma, and we shall prove it here from the ordinary version. 
Note that the simple concatenation of words fails to prove superadditivity of $\rho$ because two runs in $x$ and $y$ can be merged into a single one in $xy$, as it happens, in the above example,  to runs ${\tt abcabc}$ of $uw$ and ${\tt cabcabc}$ of $wv$.
However, there is a logarithmic upper bound on the number of runs that can be lost in this way. Each run that is a prefix of $y$ in particular yields a prefix square of $y$, and those squares have pairwise different primitive roots. By \cite[Theorem 8]{prefixsquares}, there is less than $\log_\phi|y|$ of such squares, where $\phi$ is the golden ratio.   Therefore \[\rho(n)+\rho(m) -\log_\phi n \leq \rho(m+n) \] for any $n\leq m$. Since $\log_\phi n \in o(n)$, this inequality is sufficient for the existence of the limit. This can be seen  by considering the modified sequence $\rho'(n):=\rho(n) - \log_\phi n -2$.
Then, for $n\leq m$, we obtain
\begin{align*}
\rho'(n)+\rho'(m) &=\rho(n)+\rho(m) - \log_\phi n - \log_\phi m -4 \leq \\
&\leq \rho(m+n) - \log_\phi m -4 
\\ &= \rho(m+n) - \log_\phi(2m) + \log_\phi 2-4
\\ &<  \rho(m+n) - \log_\phi(m+n) -2 =\rho'(m+n).
\end{align*}
Therefore $\rho'(n)$ is superadditive, and the convergence of $\rho'(n)/n$ establishes the convergence of $\rho(n)/n$.

\subsection{Justification of the algorithm}
We now introduce the main tool of Algorithm \ref{algor} used in Section \ref{sec3}, namely the predicate $P_d$ that captures frequency of lost positions in prefixes of a word. 
Let $d$ be a real number and let $w$ be a binary word. Let $p_1< p_2< \cdots < p_k$ be all lost positions of $w$. 
Then $P_d(w)$ if and only if  $p_i-1\geq jd$ for all $i=1,2,\dots,k$.
Informally, $P_d(w)$ means that, for all $1\leq i \leq k$, the average distance between the first $i$ lost positions (including the distance between the first lost position and the position one) in $w$ is at least $d$.

The next lemma expresses the fundamental property of $P_d(w)$ which allows to disprove $P_d(w)$ by checking prefixes of $w$.
\begin{lemma}\label{prefix}
If $P_d(w)$, then $P_d(v)$ for each prefix $v$ of $w$.
\end{lemma} 
  \begin{proof}
	Follows directly from Lemma \ref{lem}.
	\end{proof}

For each $d\in \mathbb R$, we define $N_d\in \mathbb N\cup \{\infty\}$ by
\[N_d=\max\{|w| \ |\ P_d(w)\}.\]
The value of $N_d$ is finite if and only if there are only finitely many words satisfying $P_d$.

\begin{theorem}\label{thm}
Let $N_d$ be finite. Then $\lim \rho(n)/n< 1-1/d$.
\end{theorem}
\begin{proof}
Let $w$ be a word not satisfying $P_d(w)$, and let $j\geq 1$ be the smallest integer such that $p_j-1 < jd$, where $1 < p_j$ is the $j$th lost position of $w$. Suppose that $p_j>N_d$, and let $v$ be the prefix of $w$ of length $N_d+1$. Then $P_d(v)$ by Lemma \ref{lem}. Since this contradicts the definition of $N_d$, we have proved that $p_j \leq N_d$. 
 
Since $p_j, j\leq N_d$, there are only finitely many possible values of $(p_j-1)/j$. This implies that $p_j-1\leq j(d - \varepsilon_d)$, for some $\varepsilon_d>0$ which depends only on $d$, not on $w$.

Let $w=a_0w_1a_1w_2$, where $a_0$ and $a_1$ are letters and $\abs{a_0w_0a_1}=p_j$. Repeating the same argument for $a_1w_2$, and using Lemma \ref{lem}, we obtain inductively  
 a factorization 
\[
w=a_0w_1a_1w_2a_2\cdots w_{m}a_m w_{m+1}
\]
such that $a_i$, $i=0,1,\dots,m$, are letters, and for each $i=1,2,\dots,m$, we have that 
\begin{itemize}
	\item $\abs{a_{i-1}w_i}< N_d$,
	\item $\abs{a_0w_1a_1\cdots a_{i-1}w_ia_i}$ is a lost position of $w$,
	\item there are at least $\left(\abs{a_{i-1}w_ia_i}-1\right)/\left(d-\varepsilon_d\right)$ lost positions $p$ of $w$ satisfying
	\[\abs{a_0w_1a_1\cdots a_{i-2}w_{i-1}a_{i-1}}<p \leq \abs{a_0w_1a_1\cdots a_{i-1}w_ia_i},\]
	\item $\abs{a_m w_{m+1}}\leq N_d$.
\end{itemize}

The above factorization is defined for $\neg P_d(w)$. However, such a factorization with 
$m=0$ and $w=a_0w_1$ exists also if $P_d(w)$.

Suppose now, contrary to the claim, that $\lim \rho(n)/n \geq 1-1/d$. Then there are infinitely many words $w$
such that $\rho(w)/|w|>1-1/(d-\varepsilon_d/2)$, where $\rho(w)$, by a slight abuse of notation, denotes the number of runs in $w$. The above factorization implies that $w$ contains at least 
	\[
	\sum_{j=1}^m \frac {\abs{a_{j-1}w_ja_j}-1}{d-\varepsilon_d} \geq \frac{\abs{w}-N_d}{d-\varepsilon_d}
	\]
	lost positions, which implies
	\[
	\frac{\rho(w)}{\abs w} \leq 1 - \frac 1 {\abs w}\cdot\frac{\abs{w}-N_d}{d-\varepsilon_d}\,.
	\]
	Therefore
	\[ 1- \frac 1{d-\varepsilon_d/2} <  1 - \frac 1 {d-\varepsilon_d} \cdot\frac{\abs{w}-N_d}{\abs w}\,.\]
	This inequality imposes an upper bound on the length of $w$, which is a contradiction.
\end{proof}

\section{Computation}\label{sec3}
Theorem \ref{thm} allows to compute an upper bound on $\rho(n)/n$ by a simple search described by Algorithm \ref{alg}.  By symmetry, we consider words starting with $\0$ only. 
The search space is reduced by the fact that $\neg P_d(w)$ allows to cut off all words starting with $w$ by Lemma \ref{prefix}. This is a significant help since the search performed in \cite{beyond} is burdened by suffixes which contain many open positions.

The search space can be seen as a binary tree where leaves are prefix-minimal words violating $P_d$. Algorithm \ref{alg} looks through the tree in lexicographic order starting with the word $\0$.  For example,  the node $\0\1$ has number 381\,978\,887\,301 if $d=19.3$. The node following lexicographically a leaf $w$, that is, the lexicographically next word prefix-incomparable with $w$, is $w'\,\1$ where $w=w'\0\1^i$ for some $i\geq 0$.
We illustrate how the computation ends. Let $d>2$ be such that $N_d$ is finite (in fact, then $d>18$, see below). Since $\0\1\1\0$ does not contain any lost position, $P_d(\0\1\1\0)$ trivially holds.
Let $\0\1\1\0 u$ be the lexicographically maximal word starting with $\0\1\1\0$ and satisfying $P_d$. Then the last three words checked by the algorithm are $\0\1\1\0 u \0$, $\0\1\1\0 u\1$ and $\0\1\1\1$, none of them satisfying $P_d$: we have $\neg P_d(\0\1\1\0 uc)$ by assumption on $\0\1\1\0 u$, and $\neg P_d(\0\1\1\1)$ because the position 3 of $\0\1\1\1$ is lost. The next word would be $\1$, but the algorithm considers only words starting with $\0$.
 
Not surprisingly, the size of the tree grows very quickly when $d$ gets smaller as shown in Table \ref{tab}, which summarizes results of the computation.

\begin{table}[!htb]
  \caption[]{Search results}
  \label{tab}
   \centering{
    \begin{tabular}{|r|r|r|}
		\hline
		& & \\[-2.3ex]
      $d$ \hspace{0.3ex} & $N_d$ \hspace{0.3ex} & size of the tree \hspace{0.2ex} \\ \hline
							& & \\[-1.5ex]
					$24$ & 308 & $24\, 279\, 243$ \\[0.1ex]		
					$23$ & 342 & $76\, 363\, 113$ \\[0.1ex]
					$22$ & 398 & $347\, 983\, 507$ \\[0.1ex]
					$21$ & 501 & $2\, 707\, 920\, 449$ \\[0.1ex]
					$20$ & 701 & $51\, 127\, 033\, 629$ \\[0.1ex]
					$19.8$ & 755 & $96\, 211\, 433\, 401$ \\[0.1ex]
					$19.7$ & 790 & $142\, 036\, 768\, 311$ \\[0.1ex]
					$19.5$ & 900 & $375\, 398\, 516\, 621$ \\[0.1ex]
					$19.4$ & 952 & $576\, 073\, 931\, 783$\\[0.1ex]
					$19.3$ & 1025 & $1\, 010\, 811\, 174\, 607$\\
				\hline	
	    \end{tabular}
   }
\end{table}

For $d=19.3$, we have obtained $N(19.3)=1025$, which implies
\[\lim \frac{\rho(n)}n < \frac {183}{193} \approx 0.9481865\dots \]  

The best lower bound from \cite{jamie} corresponds to $d=18.04263\dots$. Therefore, it remains to close the gap between $19.3$ and $18.04$.
It is likely that the actual value is very close to the lower bound. One indication of this is the analysis of the longest word $w_{\tt MAX}$ satisfying $P_{19.3}$ (see Table \ref{word}). Note that $w_{\tt MAX}[32..34]=\0\0\0$, which is the unique occurrence of $\0\0\0$ in $w_{\tt MAX}$. This means that the position 33 is lost and the position 32 is open. One can expect that the position 32 cannot be reasonably charged, that is, without introducing many lost positions. If this is true, then the word  $w_{\tt MAX}$ is actually not very good with respect to the prefix density of lost positions: if the position 32 gets lost, then the average distance of the two first lost positions is $16$. Therefore, instead of witnessing possibility of words with greater number of runs, it rather seems that the word $w_{\tt MAX}$ shows the need to improve the present method.   

\begin{table}[!htb]
  \caption[]{The longest word $w_{\tt MAX}$ satisfying $P_{19.3}$}
  \label{word}
   \centering{
    \begin{tabular}{l}
\0\0\1\0\1\0\0\1\0\1\1\0\1\0\0\1\0\1\0\0\1\0\1\1\0\1\0\0\1\0\1\0\0\0\1\0\0\1\0\1\0\0\1\0\1\1\0\1\0\0\\
\1\0\1\0\0\1\0\1\1\0\0\1\0\1\0\0\1\0\1\1\0\1\0\0\1\0\1\0\0\1\0\1\1\0\1\0\0\1\0\1\0\0\1\0\1\1\0\0\1\0\\
\1\0\0\1\0\1\1\0\1\0\0\1\0\1\0\0\1\0\1\1\0\1\0\0\1\0\1\1\0\0\1\0\1\0\0\1\0\1\1\0\1\0\0\1\0\1\0\0\1\0\\
\1\1\0\0\1\0\1\0\0\1\0\1\1\0\1\0\0\1\0\1\0\0\1\0\1\1\0\1\0\0\1\0\1\1\0\1\0\0\1\0\1\0\0\1\0\1\1\0\1\0\\
\0\1\0\1\1\0\1\0\1\1\0\1\0\0\1\0\1\1\0\1\0\0\1\0\1\0\0\1\0\1\1\0\1\0\0\1\0\1\1\0\1\0\0\1\0\1\0\0\1\0\\
\1\1\0\1\0\0\1\0\1\1\0\1\0\1\1\0\1\0\0\1\0\1\1\0\1\0\0\1\0\1\1\0\1\0\1\1\0\1\0\0\1\0\1\1\0\1\0\0\1\0\\
\1\0\0\1\0\1\1\0\1\0\0\1\0\1\1\0\1\0\1\1\0\1\0\0\1\0\1\1\0\1\0\0\1\0\1\1\0\1\0\1\1\0\1\0\0\1\0\1\1\0\\
\1\0\1\1\0\0\1\0\1\1\0\1\0\1\1\0\1\0\0\1\0\1\1\0\1\0\1\1\0\0\1\0\1\1\0\1\0\0\1\0\1\1\0\1\0\1\1\0\1\0\\
\0\1\0\1\1\0\1\0\1\1\0\0\1\0\1\1\0\1\0\1\1\0\1\0\0\1\0\1\1\0\1\0\1\1\0\1\0\0\1\0\1\1\0\1\0\1\1\0\0\1\\
\0\1\1\0\1\0\0\1\0\1\1\0\1\0\1\1\0\1\0\0\1\0\1\1\0\1\0\1\1\0\0\1\0\1\1\0\1\0\1\1\0\1\0\0\1\0\1\1\0\1\\
\0\1\1\0\0\1\0\1\1\0\1\0\0\1\0\1\1\0\1\0\1\1\0\1\0\0\1\0\1\1\0\1\0\1\1\0\0\1\0\1\1\0\1\0\1\1\0\1\0\0\\
\1\0\1\1\0\1\0\1\1\0\1\0\0\1\0\1\1\0\1\0\1\1\0\0\1\0\1\1\0\1\0\1\1\0\1\0\0\1\0\1\1\0\1\0\1\1\0\0\1\0\\
\1\1\0\1\0\0\1\0\1\1\0\1\0\1\1\0\1\0\0\1\0\1\1\0\1\0\1\1\0\0\1\0\1\1\0\1\0\1\1\0\1\0\0\1\0\1\1\0\1\0\\
\1\1\0\1\0\0\1\0\1\1\0\1\0\1\1\0\0\1\0\1\1\0\1\0\0\1\0\1\1\0\1\0\1\1\0\1\0\0\1\0\1\1\0\1\0\1\1\0\0\1\\
\0\1\1\0\1\0\1\1\0\1\0\0\1\0\1\1\0\1\0\1\1\0\0\1\0\1\1\0\1\0\0\1\0\1\1\0\1\0\1\1\0\1\0\0\1\0\1\1\0\1\\
\0\1\1\0\0\1\0\1\1\0\1\0\1\1\0\1\0\0\1\0\1\1\0\1\0\1\1\0\1\0\0\1\0\1\1\0\1\0\1\1\0\0\1\0\1\1\0\1\0\1\\
\1\0\1\0\0\1\0\1\1\0\1\0\1\1\0\0\1\0\1\1\0\1\0\0\1\0\1\1\0\1\0\1\1\0\1\0\0\1\0\1\1\0\1\0\1\1\0\0\1\0\\
\1\1\0\1\0\1\1\0\1\0\0\1\0\1\1\0\1\0\1\1\0\1\1\0\1\0\1\1\0\1\0\0\1\0\1\1\0\1\0\1\1\0\1\0\0\1\0\1\1\0\\
\1\0\1\1\0\1\1\0\1\0\1\1\0\1\0\0\1\0\1\1\0\1\0\1\1\0\1\0\0\1\1\0\1\0\1\1\0\1\0\0\1\0\1\1\0\1\0\1\1\0\\
\1\0\0\1\1\0\1\0\1\1\0\1\0\1\1\0\1\0\0\1\1\0\1\0\1\1\0\1\0\0\1\0\1\1\0\1\0\1\1\0\1\0\0\1\1\0\1\0\1\1\\
\0\1\0\0\1\0\1\1\0\1\0\1\1\0\1\0\0\1\1\0\1\0\1\1\0
	    \end{tabular}
   }
\end{table}

\begin{algorithm2e}[t]\label{algor}
  \caption{Find $N_d$}
  \label{alg}
    \KwIn{A number $d$.}
    \KwOut{$N_d$.}
$N_d =1$\;
$w\leftarrow \0$\;
\While{$w\neq\1$}
{find $p_1<p_2<\cdots p_k$ \tcp{all lost positions of $w$}
\If{$p_i-1\geq i\cdot d$ for all $i=1,2,\dots,k$}
{$N_d \leftarrow \max(|w|,N_d)$\;
	$w \leftarrow w\,\0$
 }
\Else{$w \leftarrow w'\1$ \tcp{where $w\in w'\0\1^*$}} 
}
\Return $N_d$
\end{algorithm2e}

\subsection{Some remarks on the implementation} Algorithm \ref{alg} was implemented in C++. The nontrivial part is finding all lost positions. The implementation of the search for $d=19.3$ stored arrays of positions $i$ with $\EE_v(i)=|v|$ for each prefix $v$ of the actual $w$. Then it was enough to recompute needed values like $\LL(i)$, $\ST(i)$ and $\EE(i)$ just for stored positions of the word being extended ($w$ or $w'$). This space consuming approach allows to speed up the search about two times. The resulting performance for $d=19.3$ was on average roughly $4.7\cdot 10^5$ nodes of the search tree per second on i5-3330 3.00GHz RAM 4GB. The bound $22/23$ of \cite{beyond} is obtained in less then 2 minutes.

\section{Final remarks and further research}
We consider only binary words in this paper. The theory of Section \ref{theory} could be adapted for a general alphabet. However, the computer search, yielding the bound which represents the main result of this paper, is done for the binary alphabet anyway. 
A better justification for the choice would be the claim that the maximum density of runs is achieved by binary words. A result of this kind was recently  obtained for a related problem of the density of distinct squares (see \cite{Manea_2015}).  In view of the presented theory, such a claim sounds reasonable also for runs because more letters imply more lexicographic orders and more Lyndon roots, thus we can expect  more lost positions. However, surprisingly, there is no such result available in the literature so far and we do not see any obvious way how to convert the intuitive argument into a formal one. For sake of simplicity, we therefore assume the binary alphabet throughout the paper and point out this problem as an interesting topic for further research. 

This paper is concerned with the asymptotic behavior of $\rho(n)/n$ and does not address questions about its progress. In particular, we cannot give a definitive answer to the question whether there exists a word $w$ with more than $\frac {183}{193}|w|$ runs. 
However, we conjecture that the asymptotic upper bound is never reached. This conjecture is based on the fact that any finite word contains some open positions (for example, the first position of the last block of letters) which further decreases the number of charged positions, thus the number of runs. 
The question is related to our discussion in Section \ref{conv}. The answer would be trivial, were the function $\rho(n)$ superadditive after all, because then $\rho(k|w|)\geq k\rho(|w|)$ would contradict the asymptotic bound.  

\section{Acknowledgments}
I thank Tomohiro I, who provided me with an implementation of the search used in \cite{beyond}, Petr Pit\v rinec who helped me to write the first version of the code, and Milan Boh\'a\v cek who gave me some valuable advise on how to make the computation fast (including not to be afraid of using global variables). 

I am also grateful to referees who helped to improve the presentation and asked 
pertinent questions leading to the section on further research.

\end{document}